\documentclass[12pt]{iopart}
\usepackage{iopams}
\usepackage{xcolor}
\usepackage{hyperref}

\expandafter\let\csname equation*\endcsname\relax
\expandafter\let\csname endequation*\endcsname\relax
\usepackage{amsmath}
\usepackage{amsthm}
\usepackage{graphicx}
\usepackage{float}
\usepackage{subcaption}

\newtheorem{theorem}{Theorem}
\newtheorem{corollary}{Corollary}
\newtheorem{lemma}{Lemma}

\begin{document}

\title[Wave functions of the Hydrogen atom in the momentum representation]{Wave functions of the Hydrogen atom in the momentum representation}

\author{M. Kirchbach$^a$ and J.\ A.\ Vallejo$^b$}

\address{$^a$Instituto de Fis{\'{i}}ca \\
Universidad Aut\'onoma de  San Luis Potos\'i (UASLP)\\
Av. Manuel Nava 6, SLP 78290, M\'exico\\
$^b$ Departamento de Matem\'aticas Fundamentales\\
Universidad Nacional de Educaci\'on a Distancia (UNED)\\
C. Juan del Rosal 10, CP 28040, Madrid, Spain}
\ead{mariana@ifisica.uaslp.mx,jvallejo@mat.uned.es}

\begin{abstract}
We construct the integral transform passing from the space representation to the
radial momentum representation for the Hydrogen atom. The resulting  wave functions
are explicitly given in terms of complex finite expansions of Gegenbauer functions
of the first and second kind, or in terms of (elementary) trigonometric functions.
We show their symmetry under the $SO(4)$ group, and their equivalence with those of
Lombardi and Ogilvie in \cite{Ogilvie}.
\end{abstract}

\vspace{2pc}

\section{Introduction}\label{secintro}

The aim of this paper is to study the construction of the wave functions of the Hydrogen atom
in the momentum representation, a topic which has been considered in \cite{PdlPlng}- \cite{Ogilvie}, among others.
To understand why this is a non-trivial exercise in basic
quantum mechanics, we need to recall some physical and mathematical ideas that will be
important in the following sections. However, as we are mainly focused on physical 
applications, we only sketch the mathematical details.

In the canonical quantization procedure pioneered by Dirac, one starts with a classical
Hamiltonian system which, for simplicity, we will assume to be $n-$dimensional and
having a phase space equal to $\mathbb{R}^{2n}$ (that is, we avoid constraints). More
generally, we could think of a $2n-$dimensional cotangent manifold $T^*M$, where $M$ is a 
configuration space.
In any case, the phase space is endowed with a canonical symplectic form $w$ which allows
the definition of things such as Hamiltonian vector fields $X_f$ associated to classical 
observables $f\in\mathcal{C}^\infty (T^*M)$ (that vector field is uniquely defined by the 
condition $w(X_f,\cdot )=-\mathrm{d}f$, where ${\mathrm d}$ is the exterior differential), and the
Poisson bracket of two classical observables $f,g\in \mathcal{C}^\infty (T^*M)$ (defined as
$\{f,g\}=w(X_f,X_g)$). Those systems admit a privileged system of so-called Darboux 
coordinates $(\mathbf{q},\mathbf{p})=(q^1,\cdots,q^n,p_1,\cdots,p_n)$, in which the
symplectic form becomes $w=\mathrm{d}q^i\wedge\mathrm{d}p_i$, equivalently, the Poisson
brackets between the coordinates (considered as functions on the phase space) are the
canonical ones: $\{q^j,p_k\}=\delta^j_k$.

Quantization is achieved by replacing the classical observables $f\in\mathcal{C}^\infty (T^*M)$
by self-adjoint operators $\hat{f}$ defined on some dense subspace of the Hilbert space of
square-integrable functions $L^2(M,\mathrm{d}\mu)$, with respect to some measure 
$\mathrm{d}\mu$ (to be constructed in the process), and the Poisson bracket by $i\hbar$ 
times the commutator 
between operators, $[\hat{f},\hat{g}]=\hat{f}\circ\hat{g}-\hat{g}\circ\hat{f}$. 
Thus, the Poisson bracket of the Darboux coordinates now becomes the well-known Heisenberg 
canonical commutation relation 
\begin{equation}\label{eq01}
[\hat{q}^j,\hat{p}_k]=i\hbar \delta^j_k I\,,
\end{equation}
where $I$ stands for the identity operator. Notice that this relation has to be valid only
in the common domain where it makes sense, and this could be the empty set \cite{Folland}.

Now, one can pose the question of what operators could be those $\hat{q}^j,\hat{p}_k$. Any
textbook on quantum mechanics presents an heuristic argument leading to the consideration of
the multiplication operator
\begin{equation}\label{eq02}
\hat{q}^j\psi(\mathbf{q})=q^j\psi(\mathbf{q})\,,
\end{equation}
and the differentiation operator
\begin{equation}\label{eq03}
\hat{p}_k\psi(\mathbf{q})=-i\hbar\frac{\partial\psi(\mathbf{q})}{\partial q^k}\,.
\end{equation}
That correspondence works well when the classical configurations space is $M=\mathbb{R}^n$, 
so the  Darboux coordinates are the Cartesian ones on $T^*M=\mathbb{R}^{2n}$, and the 
corresponding Hilbert space is $L^2(\mathbb{R}^n,\mathrm{d}^nx)$, the space of square-integrable 
functions with respect to the Euclidean measure. In that case, both
$\hat{q}^j$ and $\hat{p}_k$ are well-defined operators on the domain consisting, for instance,
of Schwartz-class functions $\mathcal{S}(\mathbb{R}^n)$, and admit unique self-adjoint
extensions (with which they are identified, respectively). 
The uniqueness of the representation \eref{eq02} and \eref{eq03} is then
guaranteed by the Stone-Von Neumann theorem: any other pair of operators satisfying
\eref{eq01}, are unitarily equivalent to \eref{eq02}, \eref{eq03}.

However, it is important to understand that things are drastically different if the 
configuration space is not all of $\mathbb{R}^n$, or even if this is the case but we choose 
coordinates different from the canonical Darboux (Cartesian) ones \cite{Essen}. This remark will be crucial 
later on.

A fundamental result in the theory of Hilbert spaces is the spectral theorem for self-adjoint 
operators, which basically states that given one of these, we can find a unitary 
equivalence in the Hilbert space such that the transformed operator acts as a multiplication 
operator, that is, the transformed operator is ``diagonal''. Each choice of an operator and the 
associated equivalence that diagonalizes it, is called a
\emph{representation} of the system. In the light of this result, we can say that the canonical
Dirac correspondence \eref{eq02}, \eref{eq03} determines the position representation, as the
operators $\hat{q}^j$ are multiplication operators. Of course, we could ask what is the momentum representation, that is: What is the unitary equivalence 
$\mathcal{F}:L^2(\mathbb{R},\mathrm{d}x)\to L^2(\mathbb{R},\mathrm{d}x)$ such that given the
operator $\hat{p}=-i\hbar\mathrm{d}/\mathrm{d}x$, the transformed operator
$\mathcal{F}\circ \hat{p}\circ\mathcal{F}^{-1}$ acts via the multiplication by identity? (We 
are considering the one-dimensional case for ease of notation).
The answer is, of course, that $\mathcal{F}$ is the well-known Fourier transform, defined as
\begin{equation}\label{eq04}
\mathcal{F}(\psi)(p)=\int_{\mathbb{R}}e^{-i\frac{p\cdot x}{\hbar}}\psi(x)\,\mathrm{d}x\,.
\end{equation}
It suffices to note that the desired result is equivalent to
\begin{equation}\label{eq05}
\mathcal{F}(\hat{p}\psi)(p)=p\cdot (\mathcal{F}\psi)(p)\,.
\end{equation}
This is proved by observing that unitary operators preserve the density of the domain and the 
closures, and integrating by parts (since the functions $\psi(x)$ in the domain of $\hat{p}$, 
vanish at infinity).

Expression \eref{eq05} suggests the following heuristic motivation for the introduction
of the Fourier transform: If we think of it as an integral transform whose kernel are the
``eigenfunctions'' of the momentum operator\footnote{Of course, that only makes sense at a 
formal level: The 
spectrum of \eref{eq03} is purely continuous and it possesses no eigenfunctions in
$L^2(\mathbb{R},\mathrm{d}x)$.} \eref{eq03}, the plane waves $e^{-i\frac{p\cdot x}{\hbar}}$,
then it is clear that the derivative inside the integral transform, after integration by parts 
so that it  acts on the exponential, will give us the multiplying $p$. Then, given that the 
integration is done with respect to the $x$ variable, $p$ can be factored outside the integral. 
We will generalize this idea in later sections.

Finally, after these admittedly long preliminaries, we can formulate the problem we are 
interested in, which is to compute the radial wave functions (using polar spherical 
coordinates) of the Hydrogen atom in the momentum representation. In Section \ref{sec2}
we will recall the correct form of the radial momentum in these coordinates, and we will
compute the unitary transformation which diagonalizes it, by means of an integral transform
whose kernel is given by its formal eigenfunctions. In Section \ref{sec3} we briefly review
the existing literature on this subject; in particular we review the Lombardi-Ogilvie approach, 
whose output is a set of wave functions that are not manifestly $SO(4)-$invariant. In 
\ref{sec4} we give explicit expressions  for the resulting radial functions of the Hydrogen 
atom in momentum space following our proposal, which have the advantage of being 
$SO(4)-$invariant; we also prove the equivalence with the functions of Lombardi-Ogilvie. 
Finally, in Section \ref{sec5} we present some conclusions and perspectives for future 
developments.

\section{The momentum representation for the Hydrogen atom}\label{sec2}

Here we face one of the 
difficulties we mentioned earlier, namely, the construction of the momentum operator in
coordinates other than the Cartesian ones. If polar spherical coordinates $(r,\theta,\varphi)$
are chosen, then the radial momentum can not be $p'_r=-i\hbar\partial /\partial r$. The 
reason is that this operator is not even symmetric, and, as a consequence, its expectation 
values calculated on Schr\"odinger's radial wave functions for a central potential are not real \cite{Utpal}.  In fact, one gets an expression like
\begin{eqnarray*}
\langle p_r'\rangle &=& i\hbar \int {\mathrm d}\Omega \left[
\int_0^\infty r^2\psi({\mathbf r}) \frac{\partial \psi^\ast ({\mathbf r})}
{\partial r} {\mathrm d}r\right]
+ 2i\hbar \int {\mathrm d}\Omega \int _0^\infty r^2 |\psi ({\mathbf r})|^2\,,
\end{eqnarray*} 
and the usual uncertainty relations do not hold in this case. This situation leads to
unwanted consequences such as a violation of the second law of thermodynamics \cite{Esther},
so it has to be avoided.

De Witt \cite{deWitt} has devised a scheme to find the correct momentum 
operator conjugate to an arbitrary coordinate function (that is, maintaining the Heisenberg
commutation relations \eref{eq01}): Given any $n-$dimensional Riemannian 
manifold $(M,g)$ (which in our case will be the Euclidean flat $3-$space $M=\mathbb{R}^3$), and 
a coordinate system $x^j$ on it, with $1\leq j \leq n$ (in our case the polar spherical
coordinates $(r,\theta,\varphi)$), the momentum operator corresponding to $x^j$ is given
by
\begin{equation}\label{eq06}
\hat{p}_j=-i\hbar\frac{\partial}{\partial x^j}-\frac{1}{2}i\hbar\sum^n_{k=1}\Gamma^k_{kj}\,,
\end{equation}
where $\Gamma^k_{jl}$ are the Christoffel symbols of the metric $g$ expressed in the $x^j$
coordinate system. In the case of the flat three-dimensional Euclidean metric in polar 
spherical coordinates, it is immediate to find that the Christoffel symbols of interest
are $\Gamma^\theta_{\theta r}=1/r$ and $\Gamma^\varphi_{\varphi r}=1/r$, so we get the
following expression for the radial momentum operator:
\begin{equation}\label{eq07}
\hat{p}_r=
-i\hbar\left( \frac{\partial}{\partial r}+\frac{1}{r}\right) =
-i\hbar\frac{1}{r}\frac{\partial}{\partial r}\circ r\,,
\end{equation}
see also \cite{Essen}.
From a mathematical point of view, the operator $\hat{\mathbf{p}}$ defined in \eref{eq06}
is not free from problems. 
Its natural domain of definition are the smooth functions of compact support in the punctured 
space $\mathbb{R}^3-\{0\}$, that is, $\mathcal{C}_c(\mathbb{R}^3-\{0\})$, and on this domain
it is a symmetric operator. However, it is not an observable, since its deficiency indices
are $(0,\infty)$ and, consequently, it does not have self-adjoined extensions (see \cite{Lax}).
Although, in particular, $\hat{p}_r$ is not an observable, its square $\hat{p}^2_r$ is (since, being  non-negative, it admits a Fredrich's extension).

Now we turn our attention to the Hydrogen atom. This is a particular case of a quantum 
system of two particles interacting
through a central potential. The wave function $\psi(\mathbf{q})$ describing the relative 
motion, satisfy the Such\"dingier equation
\[
(H\psi)(\mathbf{q})=\left[ -\frac{\hbar^2}{2\mu}\Delta +V(|\mathbf{q}|) \right]\psi(\mathbf{q})
=E\psi(\mathbf{q})\,,
\]
where $V(|\mathbf{q}|)$ is a radial function (that is, invariant under rotations), and $\mu$ is 
the reduced mass. Introducing polar spherical coordinates, and doing a little algebra, this 
equation becomes
\begin{equation}\label{eq08}
\left[\frac{\hat{p}^2_r}{2\mu}+\frac{\mathbf{L}^2}{2\mu r^2}+V(r)\right]\psi(r,\theta,\varphi)=
E\psi(r,\theta,\varphi)\,,
\end{equation}
where $\mathbf{L}$ is the relative orbital angular momentum, and $V(r)=-e^2/r$. By separation 
of variables, the solution to \eref{eq08} is a superposition of functions of the form
\[
\psi_{N\ell m}(r,\theta,\varphi)=R_{N\ell}(r)Y^m_\ell(\theta,\varphi)\,.
\]

In what follows, we will be interested in the resulting equation for $R_{N\ell}(r)$ (the radial 
invention), called the radial equation:
\begin{equation}\label{eq09}
\left[\frac{\hat{p}^2_r}{2\mu}+\frac{\ell (\ell+1)\hbar^2}{2\mu r^2}+V(r)\right] R_{N\ell}(r)
=E_NR_{N\ell}(r)\,,
\end{equation}
and our goal is: First, to find a unitary transformation 
$\mathcal{H}:L^2(]0,\infty [\,,r^2\mathrm{d}r)\to 
L^2(\mathbb{R},r^2\mathrm{d}r)$\footnote{Although $r$ lives in $]0,\infty[\,$, the radial momentum $p_r$ takes real values,
including negative ones.}
such that, given the operator $\hat{p}_r$, the transformed operator
$\mathcal{H}\circ \hat{p}_r\circ\mathcal{H}^{-1}$ acts via the multiplication by identity, that 
is, to find the momentum representation. Actually, we will not make use of the inverse transform, so we will prove \eqref{eq05} directly.
Second, we want to express the radial wave function
in this representation. We devote the remaining of this section to the first task, and the
remaining sections to the second.

\noindent
In order to find the transformation $\mathcal{H}$, we follow the heuristic idea exposed after
\eref{eq05}, and start by finding the (formal) eigenfunctions of $\hat{p}_r$, that is, for each $p_r\in \,]0,\infty [\,$, we want the solutions to  
$\hat{p}_r\psi_{p_r}(r) =\pm p_r\psi_{p_r}(r)$. As this equation is separable, we readily get 
\begin{equation}\label{eq10}
\psi^{\pm}_{p_r}(r)=\frac{e^{\pm i\frac{p_r\cdot r}{\hbar}}}{r}\,,
\end{equation}
which represent outgoing and incoming spherical waves.

\noindent\textbf{Remark:} \emph{In what follows, we will make free use of natural units for which
$\hbar=1$, for the sake of readability}.

Notice that, outgoing and incoming spherical waves can be 
expressed in terms of spherical Hankel functions of the first $h^{(1)}_0(x)$ and second $h^{(2)}_0(x)$ kind,  respectively given by  
\[
h^{(1)}_0(x)=-i\frac{e^{ix}}{x}\,, \quad h_0^{(2)}(x)=i\frac{e^{-ix}}{x}\,.
\]
For instance, we have:
\[
\psi^{-}_{p_r}(r)=\frac{e^{- ip_r r}}{r}=-ip_rh^{(2)}_0(p_r r)\,.
\]

Now, we can state the main result of this section. For the sake of clarity, 
we drop the subindex $r$ in the radial momentum, writing
simply $p$ instead of $p_r$. 

\begin{theorem}\label{theorem1}
Define the transformation 
$\mathcal{H}:L^2(]0,\infty [\,,r^2\mathrm{d}r)\to L^2(\mathbb{R},r^2\mathrm{d}r)$ by
\begin{equation}\label{eq11}
(\mathcal{H}\phi)(p)=\int^\infty_0 \phi(r)\psi^-_{p}(r)r^2\mathrm{d}r\,.
\end{equation}
Then, it holds
\begin{equation}\label{eq12}
\mathcal{H}(\hat{p}\phi)(p)=p\cdot (\mathcal{H}\phi)(p)\,.
\end{equation}
\end{theorem}

\begin{proof} First of all, notice that $\mathcal{H}$ acts as
\[
(\mathcal{H}\phi)(p)=-ip\int^\infty_0 \phi(r)h^{(2)}_0(pr)r^2\mathrm{d}r=
\int^\infty_0 \phi(r)e^{-ipr}r\mathrm{d}r\,.
\]
Written in this form, it is clear that it can be considered as a Fourier 
transform on a weighted Lebesgue space $L^1(]0,\infty[\,,r\mathrm{d}r)$, as the weight 
$u(r)=r$ is a locally integrable, positive function on $]0,\infty[\,$\footnote{Yet another 
point of view is that $(\mathcal{H}\phi)(p)$ is the Fourier transform of the Borel 
measure $\phi(r)r\mathrm{d}r$ on $]0,\infty[\,$. The appearance of the Hankel function 
$h^{(2)}_0$ is no surprising at all, since it is well known that the ordinary Fourier transform 
of a radial function is given by a Hankel transform.}. It is also apparent that the Schwartz
class $\mathcal{S}(]0,\infty[\,)$ is included in the domain of $\mathcal{H}$ (notice that
the vector space of Schwartz functions is closed under multiplication by polynomials), so
we can work in this subspace (or even in the smaller subspace of functions with compact 
support) and apply a standard density argument to extend the definition
to the broader space $L^2(]0,\infty [\,,r^2\mathrm{d}r)$. The proof that $\mathcal{H}$ is
unitary is also standard. Finally, we integrate by parts to show that \eref{eq05} holds
for any $\phi$ of compact support on $]0,\infty[\,$:
\begin{align*}
\mathcal{H}(\hat{p}\phi)(p) &= 
\int^\infty_0 -\frac{i}{r}\frac{\partial}{\partial r}(r\phi(r))\psi^-_{p}(r)r^2\,\mathrm{d}r \\
&= \int^\infty_0 -i\frac{\partial}{\partial r}(r\phi(r))\psi^-_{p}(r)r\,\mathrm{d}r \\
&=\int^\infty_0 r\phi(r)\left( i\frac{\partial}{\partial r}\right)(r\psi^-_{p}(r))\,\mathrm{d}r \\
&= -\int^\infty_0\phi(r)\left( -\frac{i}{r}\frac{\partial}{\partial r}\right)(r\psi^-_{p}(r))\,r^2\mathrm{d}r \\
&= -\int^\infty_0\phi(r)(\hat{p}\psi^-_{p})(r)r^2\,\mathrm{d}r \\
&=p\int^\infty_0\phi(r)\psi^-_{p}(r)r^2\,\mathrm{d}r \\
&= p(\mathcal{H}\phi)(p)\,.
\end{align*}
\end{proof}

\section{H Atom  wave functions in the momentum representation: Previous results}\label{sec3}

Before considering the problem of writing the solutions to \eref{eq09} in the momentum
representation, given by the transform $\mathcal{H}$ in \eref{eq11} of the
preceding section, let us consider previous approaches to this problem.

\subsection{The functions of Podolsky and Pauling}\label{sec3.1}

Podolsky and Pauling \cite{PdlPlng} (and later also Hylleraas \cite{Hylleraas}) tried to find 
solutions $G_{N\ell m}\left(\mathbf{p}\, \right)$ to \eref{eq09} by considering the standard  Fourier transform as the 
unitary transformation that diagonalizes it, thus arriving at the formulae,
\begin{align*}
G_{N\ell m}(\mathbf{p}\, )&=\frac{1}{(2\pi\hbar )^{\frac{3}{2}}}\int_0^\infty  
e^{ \frac{i}{\hbar} \mathbf{p}\cdot \mathbf{r}} {\mathcal R}_{N\ell m} (\mathbf{r}\, ){\mathrm d}\mathbf{r}\\
&= {\mathcal G}_{N\ell }(p)Y_\ell ^m(\theta_p,\varphi_p),
\end{align*}
where ${\mathcal R}_{N\ell m}(\mathbf{r}\, )$ denote the usual wave functions of the Hydrogen
atom in the position representation, and $Y_\ell ^m(\theta_p,\varphi_p)$ are the well known spherical harmonics. 

One method for evaluating the above integral consists in expanding the plane wave as
\begin{equation*}
e^{i\frac{p}{\hbar} r\cos\theta }=
\sum_0^\infty \sum_{m=-\ell'}^{m=\ell'}i^{\ell'} (2\ell' +1)j_{\ell'} (pr)P_{\ell'}(\cos\theta)\,,
\label{Raylay}
\end{equation*}
which leads to the Fourier transform expanding into a series of integral transforms whose kernels are spherical Bessel functions of every order, and to a factorization in terms of the variables $pr$ and $\cos\theta$.
Had we considered instead the radial component of the momentum as a variable, naively taken as the classical value  $p_r=p\cos\theta$, we would have arrived at
\begin{equation}
e^{i\frac{p_r}{\hbar}r}  =
-i\frac{p_rr}{\hbar}\left[j_0\left(\frac{p_r}{\hbar }r \right)-in_0\left(\frac{p_r}{\hbar }r \right) \right]\,,
\end{equation}
that is, to a kernel with spherical Bessel functions of the first and second kind, and of zeroth order. We discuss this in more detail after eq.~(\ref{Olg_33}) below. 
Irrespective of the techniques applied, the result is known, and here
we quote reference \cite{Szmytkowski}, where the 
radial wave functions in the momentum representation 
have been worked out for the general case of any dimension. Atoms in higher dimensions have unusual properties that can serve as signature of extra dimensions in quantum mechanics, as they are predicted by grand unification and string theories.  For the three dimensional case
at hand, the radial wave functions ${\mathcal G}_{N\ell }(p)$ turn out to be 
\begin{equation}
{\mathcal G}_{N\ell }(p)=(2\hbar \beta)^{\frac{5}{2}}\Gamma(\ell +1)
\sqrt{  \frac{ (N-\ell -1)!N!}{\pi (N+\ell )!}}
\frac{(4\hbar \beta p)^\ell}{(\hbar^2\beta^2 +p^2)^{\ell +2}}C^{\ell +1}_{N-\ell -1}\left(\frac{\hbar^2 \beta^2-p^2}{\hbar^2 \beta^2+p^2}\right)\,.
\label{Szmytk}
\end{equation}
Here, $\beta=Z\mu c\alpha /N\hbar$ is a constant given in terms of the fine structure constant 
$\alpha$, the reduced mass $\mu$, and the atomic number $Z$, so $\hbar \beta $ equals  $q_n$  
in eq.~(4.7) in \cite{Szmytkowski}. The symbols $C_n^\alpha(x)$ stand for Gegenbauer polynomials. After
a change of variable in their argument, given by
\begin{equation}
\cos\chi_p =\frac{\hbar^2 \beta^2-p^2}{\hbar^2\beta^2+p^2}, \quad \sin\chi_p =\frac{2\hbar \beta p}{\hbar^2\beta^2 +p^2}, \quad 
\chi_p \in \left[0,\pi  \right]\,,
\label{Cnvrs_cs}
\end{equation}
the wave functions in the momentum representation obtained in this way, known as Podolsky-Pauling functions, read
\begin{eqnarray}
{\mathcal G}_{N\ell }(\chi_p)
&=&(2\hbar \beta) ^{\frac{1}{2}} 2^{\ell+2}
\sqrt{  \frac{ (N-\ell -1)!N!}{\pi (N+\ell )!}}
\Gamma(\ell +1) \cos^2\frac{\chi_p}{2}{\mathcal S}_{(N-1)\ell }(\chi_p),\nonumber\\
{\mathcal S}_{(N-1)\ell }(\chi_p)&=&\sin^\ell \chi_p C_{N-\ell -1}^{\ell +1}(\cos \chi_p).
\label{Ultra_Sphr_HAt}
\end{eqnarray}
Here, ${\mathcal S}_{(N-1)\ell }(\chi_p)$ coincides with the second polar-angle part of the ultra-spherical harmonics, 
$Y_{K\ell m }(\chi_p, \theta_p,\varphi_p)={\mathcal S}_{K\ell }(\chi_p)Y_\ell ^m(\theta_p,\varphi_p)$, with  $K=N-1$ denoting the four-dimensional angular momentum value. 

The great advantage of these functions is that they determine bases of the irreducible 
representations of the orthogonal group $SO(4)$.  Indeed, it is clearly possible to factor
the ultra-spherical harmonics $Y_{(N-1)\ell m}(\chi_p,\theta_p,\varphi_p)$ in the expression
for the full wave functions $G_{N\ell m}(\mathbf{p}\, )$, thus explicitly showing the
full $O(4)$ symmetry of the Hydrogen atom in the momentum representation.

It is important to point out that the Podolsky-Pauling functions depend on $p^2 =p_x^2+p_y^2+p_z^2$, the squared (Euclidean) norm of the linear momentum, and not on the genuine radial momentum , $p_r$, which is the  projection of  ${\mathbf p}$ on the direction
of ${\mathbf r}$, and can naively be thought in terms of its classical expression,  
$p_r=p\cos\theta$. 
In that sense, the Podolsky-Pauling functions seem conjugate to the radial Hydrogen atom wave functions when expressed in terms of $r^2=x^2+y^2+z^2$, rather than to the radial coordinate itself (see the discussion in subsection \ref{sec44}, regarding the corresponding uncertainty
relations).

\subsection{The functions of Lombardi and Ogilvie}\label{sec3.2}
 Lombardi \cite{Lombardi} and Ogilvie \cite{Ogilvie}
worked out  the  total momentum space wave function, factorized as $\alpha({\hat{\mathbf p}})=\alpha (p_r)\beta (\theta)\rho(\varphi)$, as follows:  Starting with the radial equation \eref{eq09}, they made the replacements
\begin{eqnarray}
-i\frac{1}{r}\frac{\partial}{\partial r} r &\longrightarrow&  p_r,\
\nonumber\\
\frac{1}{r}&\longrightarrow&-i\int_{-\infty}^{p_r} {\mathrm d}{p^\prime _r}:=I,
\label{2nd_rplsmnt}\\
R_{N\ell}(r)&\longrightarrow & \alpha_{N\ell}( p_r),
\nonumber
\end{eqnarray}
and the resulting second-order integral equation (expressed in natural units of $\hbar =1$ 
(see eq.~(32) in \cite{Ogilvie})),
\begin{eqnarray}
{\Big(}p_r^2-2\mu E +\ell (\ell+1)I^2  +2\mu Ze^2 I {\Big)}\alpha_{N\ell}(p_r)=0,
\label{Lmbrd1}
\end{eqnarray} 
was solved by taking its second-order derivative and converting it to 
a differential equation in the $p_r$ coordinate.
The solutions to (\ref{Lmbrd1}) are then given in eq.~(33) in \cite{Ogilvie}, and can be rewritten\footnote{For example, identifying $\beta$ here with $p_0/N$ in \cite{Ogilvie}.}, for the sake of comparison, with the notations of the preceding section, yielding
\begin{eqnarray}
\alpha_{N\ell}(p_r)&=&{\mathcal N}^L_{N\ell}\sum_{k=0}^{N-\ell-1}   c_{N\ell}^k
 \left(\frac{i\beta }{p_r -i\beta}\right)^{\ell +k+2}\nonumber\\
&=&{\mathcal N}^L_{N\ell}\sum_{k=0}^{N-\ell-1}   c_{N\ell}^k
\left(
\frac{i\beta \left( p_r+i\beta\right) }{
p_r^2 +\beta^2}\right)^{\ell +k+2}\nonumber\\[6pt]
c_{N\ell}^k&=&  \frac{2^k(N -\ell-1)! (\ell +k+1)!}{
k!(N-\ell -k-1)!(2\ell +k+1)!},
\label{Olg_33}
\end{eqnarray} 
Notice that here, following Weyls's symmetrization recipe, $
p_r=\frac{1}{2}( {\mathbf p}\cdot \frac{{\mathbf r}}{r}  +\frac{{\mathbf r}}{r}\cdot {\mathbf p})$. 
The normalization constant, here denoted by  ${\mathcal N}^L_{N\ell}$, can be consulted  in eq.~(33) in \cite{Ogilvie}.

These calculations present several problems from a mathematical point of view, the most
important one coming from the replacement \eref{2nd_rplsmnt}, which needs a careful study
of the domain in which is valid. From a physical perspective, the wave functions of Lombardi and Ogilvie, even being deduced from the correct radial momentum operator, have the
drawback of not being explicitly $SO(4)-$invariant, a fact which perhaps is the explanation
of why they are not more widespread.

In the following section, we calculate the  $\psi_{N\ell}(p_r)$ functions  from the integral transform in Theorem 1, finding two new equivalent, although different looking, representations. We will prove that these two representations are necessarily equivalent to the $\alpha_{n\ell}(p_r)$ functions above.

As to the angular part of the wave functions, they have been constructed in \cite{Lombardi}, \cite{Ogilvie}
as eigenfunctions of the operator ${\hat p}_\theta $ given by
\begin{equation}
 r{\hat p}_\theta = -i\hbar \frac{1}{\sqrt{\sin\theta }}\frac{\partial }{\partial \theta}\sqrt{\sin\theta}=-i\hbar \left( \frac{\partial }{\partial \theta }+\frac{1}{2}\cot\theta \right), \quad \theta\in [0,\pi],
\label{PthetaHerm}
\end{equation}
which is transformed by means of the unitary operator  
$t=\exp (i\pi u)$ making the change of variable $u=\cos\theta$, in which case ${\hat p}_\theta$ becomes 
${\hat p}_t=-i\partial /\partial u=\pi t\partial /\partial t$. Obviously, the angular wave
functions so obtained are essentially distinct from the associated Legendre function, $P^m_\ell (\cos\theta)$,  appearing in \eref{Ultra_Sphr_HAt}. In a way similar, the azimuthal wave 
functions are handled by setting $v=\exp(i\varphi)$, leading to the operator
${\hat p}_v=v\partial/\partial v$. In fact, the complex functions by Lombardi and Ogilvie are 
quite different from the real ones  by Podolsky and Pauling. 
In the following subsection, we briefly consider the possibility that the angular part of the 
momentum space wave functions may not be those given in the preceding works by Podolsky-Pauling
or Lombardi-Ogilvie, so other parametrizations may be eligible as well. 

\subsection
{The angular functions of Liu, Tang and Xun for a particle on the sphere}

The angular parts of the wave functions in momentum space, for a motion on a sphere, can be constructed using Dirac's theory of constraints and the associated star-bracket.
In so doing, one defines the geometric momentum, 
\begin{equation}
{\hat p\,}=-i\hbar \left( \nabla +H{\mathbf n}\right),
\end{equation}
where $\nabla $ denotes the local gradient operator on the sphere,  $H$ denotes the mean curvature of the sphere as embedded in three dimensional Euclidean space (an extrinsic measure of curvature), and ${\mathbf n \, }$ is the normal vector (eqs.~(3)-(4) in \cite{QHLiu}). Then, the most general form of its $z$ component, 
denoted by ${\hat p}_{(\alpha,\beta)z}$ in eq.~(8) in \cite{QHLiu}, is obtained as, 
 \begin{equation}
{\hat p}_{(\alpha, \beta)z} =-\frac{i\hbar }{r}\left( -\sin\theta \frac{\partial }{\partial \theta} -\alpha \cos\theta \right) +\frac{\hbar }{r}\beta \cos\theta.
\label{Liu_pz}
\end{equation}
It is further argued in \cite{QHLiu} that the values of $\alpha=1$, $\beta=0$ are particularly appropriate (see (33) there). For these parameter values, ${\hat p}_{(\alpha,\beta)z}$ takes the form,
\begin{equation}
{\hat p}_{(\alpha=1,\beta=0) z} = -\frac{i\hbar }{r}\left( -\sin\theta \frac{\partial }{\partial \theta} - \cos\theta \right)=\frac{i\hbar}{r}\sin\theta\left(
\frac{\partial }{\partial \theta } +\cot\theta \right),
\label{pz_alpha1_beta0}
\end{equation}
which is clearly distinct from (\ref{PthetaHerm}).
The corresponding eigenfunctions, reported (modulo the normalization factor) in eq.~(47) in \cite{QHLiu}, are,
\begin{equation}
\psi_z(\theta)=\frac{1}{\sqrt{2\pi\hbar}  \sin\theta}\tan ^{-ip_z} \left( \frac{\theta}{2}\right)=\frac{1}{\sqrt{2\pi \hbar}  \sin\theta }e^{-ip_z\ln \tan \frac{\theta}{2}}\,.
\end{equation}

Thus, the scheme of Dirac's second-class constrained quantization, as applied to the Hydrogen
atom, gives angular wave functions different from those in \cite{Lombardi}, \cite{Ogilvie}.
In this regard, it has to be noticed that the construction in this section depends on the particular embedding of the sphere in $\mathbb{R}^3$ one chooses (through the normal vector field $\mathbf{n}$), and it has been suggested that the corresponding
curvature effects can be detected by specially designed momentum spectrometers \cite{QuiNew,Das}.
In this sense, which functions are best suited to describe the actual Hydrogen atom is a question that
should be decided on the basis of such experiments.

\section{Radial wave functions in the momentum representation: A proposal}\label{sec4}

In this section we apply the integral transform  in Theorem 1 to the study of the radial
wave functions of the Hydrogen atom in the momentum representation. Because the wave functions 
are defined up to a phase, in the following we shall ignore the imaginary factor in front of 
$h^{(2)}_0(p_rr/\hbar)$, but we will recover the $\hbar$.

\subsection{Complex finite expansions in terms  of Gegenbauer functions}

Let $R_{Nl}(r)$ be a radial wave function, solution to \eqref{eq09}. Notice that it has a
dependence on the energy $E$ which, for bounded states, can be parametrized by an integer $N$.
In what follows, it will be necessary to make this dependence explicit, hence we will
introduce a new variable $\rho =2\beta r$, where $\beta=Z\mu c\alpha /N\hbar$ has been already
defined in subsection \ref{sec3.1}. Consider now the radial function 
$\phi (r)=R_{N\ell }(2\beta r)=R_{N\ell}(\rho)$. It is well known (see, for instance, 
\cite{Weinberg}) that it can be expressed in terms of the Laguerre polynomials $L^a_n(\rho)$ as
\begin{equation}\label{Ur1}
R_{N\ell}(\rho)=\mathcal{N}_{N\ell}e^{-\rho/2}\rho^\ell L^{2\ell +1}_{N-\ell -1}(\rho)\,,
\end{equation}
where the normalization constants $\mathcal{N}_{N\ell}$ (different from the ${\mathcal N}^L_{N\ell}$ of subsection \ref{sec3.2}) are given by
\[
\mathcal{N}_{N\ell}=(2\beta)^{3/2}\sqrt{\frac{(N-\ell -1)!}{2n\cdot (N+\ell)!}}\,.
\]

The transformed radial wave function
in momentum space, according to \eqref{eq11}, can be computed as 

\begin{eqnarray*}
\psi_{N\ell}(p_r)&=& \frac{p_r}{\hbar }\int_0^\infty h^{(2)}_0\left( \frac{p_r}{\hbar }r \right) R_{N\ell }(\rho) r^2{\mathrm d}r\nonumber\\
&=& \frac{p_r}{\hbar }
\int_0^\infty \left( j_0\left( \frac{p_r}{\hbar}r\right) -  in_0
\left(\frac{p_r}{\hbar }r \right) \right)R_{N\ell}(\rho) r^2{\mathrm d}r.
\end{eqnarray*}
We will separately study the integrals appearing here, so we define $I_1$ and $I_2$ as,
\begin{eqnarray}
I_1\left(\frac{p_r}{2\hbar \beta } \right)&=&\frac{1}{{\mathcal N}_{N\ell}}\int_0^\infty j_0\left( \frac{p_r}{\hbar }r \right)R_{N\ell }(\rho) r^2{\mathrm d}r\nonumber\\
&=&\frac{1}{(2\beta)^3}\int_0^\infty j_0\left( \frac{p_r}{2\beta \hbar }\rho \right)e^{-\frac{\rho}{2}} L_{N-\ell -1}^{2\ell +1}(\rho )\rho ^\ell \rho^2{\mathrm d}\rho,
 \label{Int_1st}
\end{eqnarray}
and 
\begin{eqnarray}
I_2\left(\frac{p_r}{2\hbar \beta } \right)&=&\frac{1}{{\mathcal N}_{N\ell}}\int_0^\infty n_0\left( \frac{p_r}{\hbar }r \right)R_{N\ell }(\rho) r^2{\mathrm d}r, \nonumber\\
&=&\frac{1}{(2\beta)^3}\int_0^\infty n_0\left( \frac{p_r}{2\beta \hbar }\rho \right)e^{-\frac{\rho}{2}} L_{N-\ell -1}^{2\ell +1}(\rho ) \rho ^\ell \rho^2{\mathrm d}\rho,
\label{Ur6}
\end{eqnarray}
respectively, and rewrite the wave function under consideration as,
\begin{eqnarray}
\psi_{N\ell}(p_r) &=&  \frac{p_r}{\hbar}{\mathcal N}_{N\ell}\left(I_1\left(\frac{p_r}{2\hbar \beta } \right)- iI_2\left(\frac{p_r}{2\hbar \beta } \right)\right).
\label{Ur7_pr}
\end{eqnarray}

In order to evaluate $I_1\left(\frac{p_r}{2\hbar \beta } \right)$ and $I_2\left(\frac{p_r}{2\hbar \beta } \right)$, we express the Hydrogen atom wave functions, based on Laguerre's generalized polynomials, in terms of finite sums of Slater-type functions. To this end,
we first make use of the well-known expression for the Laguerre polynomials,
\begin{eqnarray}
L_n^\alpha (\rho)&=&\sum_{t=0}^n  (-1)^t \left( \begin{array}{c} 
n +\alpha \\ 
n-t \end{array}\right)
\frac{\rho^t}{t!},
\label{Ur4}
\end{eqnarray}
which we then substitute into the equations ~(\ref{Int_1st}), and (\ref{Ur6}). In so doing we find, 
\begin{eqnarray}
I_1\left(\frac{p_r}{2\hbar \beta} \right)&=& \sum_{t=0}^{N -\ell -1} \binom{N+\ell}{N-\ell -1 -t} \frac{(-1)^t}{{t! (2\beta)^3}}\int_0^\infty
j_0\left( \frac{p_r}{2 \hbar \beta  }\rho \right)e^{-\frac{\rho }{2}}
\rho ^{\ell +t +2}{\mathrm d}\rho \nonumber\\
&=& \sum_{t=0}^{N -\ell -1} \binom{N+\ell}{N-\ell -1 -t} \frac{(-1)^t}{{t! (2\beta)^3}}\sqrt{\frac {\beta\hbar  \pi}
{p_r}} {\mathcal I}_1\left(\frac{p_r}{2\hbar \beta} \right),
\label{Ur7_a}
\end{eqnarray}
where ${\mathcal I}_1\left(\frac{p_r}{2\hbar \beta} \right)$ is an integral of the 
Hankel transform type (with a kernel given by the Bessel function of the first kind, $J_{1/2}(\rho )$), given by 
\begin{eqnarray*}
{\mathcal I}_1\left( \frac{p_r}{2\hbar \beta}\right)
&=& \int_0^\infty
J_{1/2}\left( \frac{p_r}{2 \hbar \beta  }\rho \right)
e^{-\frac{\rho }{2}}{\rho ^{\ell +t +\frac{3}{2}}} {\mathrm d}\rho\,.
\end{eqnarray*}
Analogously, the second integral is evaluated as,
\begin{eqnarray}
I_2\left(\frac{p_r}{2\hbar \beta} \right)&=&\sum_{t=0}^{N -\ell -1}\binom{N+\ell}{N-\ell -1 -t}  \frac{(-1)^t}{{t! (2\beta)^3}}\sqrt{\frac {\beta \hbar\pi}
{p_r}}{\mathcal I}_2\left(\frac{p_r}{2\hbar \beta} \right),
\label{Ur8_a}
\end{eqnarray}
where ${\mathcal I}_2\left(\frac{p_r}{2\hbar \beta } \right)$ is another integral of the 
Hankel transform type (with a kernel given by the Bessel function of the second kind, $Y_{1/2}(\rho)$), given by
\begin{eqnarray*}
{\mathcal I}_2\left(\frac{p_r}{2\hbar \beta} \right)
&=&\int_0^\infty
Y_{1/2}\left( \frac{p_r}{2 \hbar \beta  }\rho \right)e^{-\frac{\rho }{2}}
 {\rho ^{\ell +t +\frac{3}{2}}}{\mathrm d}\rho\,.
\end{eqnarray*}
The aforementioned Hankel transforms, 
${\mathcal I}_1\left( \frac{p_r}{2\hbar \beta }\right)$ and 
${\mathcal I}_2\left(\frac{p_r}{2\hbar \beta } \right)$, are tabulated in \cite{RyzGrad} (eq.~{\bf 6.621} on p.~699, and in {\bf 6.622} on p.~700, respectively). Therefore, we obtain
\begin{eqnarray*}
{\mathcal I}_1\left( \frac{p_r}{2\hbar \beta}\right)&=&\left(
2x\right)
^{\left(\ell +\frac{5}{2} +t\right) }\Gamma\left(\ell +3 +t \right)  P_{\ell +\frac{3}{2}+t} ^{-\frac{1}{2}}
\left(x \right)\,,
\end{eqnarray*}
and
\begin{eqnarray*}
{\mathcal I}_2\left(\frac{p_r}{2\hbar \beta} \right) &=&- \frac{2}{\pi} (2x)^{\left(\ell +\frac{5}{2} +t\right) }
\Gamma\left(\ell +3 +t \right)Q_{\ell +\frac{3}{2}+t} ^{-\frac{1}{2}}
\left(x\right)\,,
\end{eqnarray*}
respectively, where
\begin{eqnarray}
x&=&\frac{1}{2}\left[\left( \frac{1}{2}\right)^2 + \left(\frac{p_r}{2\hbar \beta}\right)^2\right]^{-\frac{1}{2}}.
\label{z_argmnt}
\end{eqnarray}
  Here,  $P_\nu^{-\mu}(x)$ and $Q_\nu^{-\mu}(x)$ denote the associated Legendre functions of the first  and the second kind, 
characterized by their order $\nu$, and degree $\mu$, with ${\mathcal R}e (\nu +\mu)>0$ (they are also known as Ferrers functions of the first and second kind \cite{Cohl}).
Thus, we arrive at the following result.
\begin{lemma}
The radial momentum wave functions $\psi_{N\ell}(p_r)$ can  be expressed as the following finite series expansion:
\begin{eqnarray*}
\psi_{N\ell}(p_r)=&&\frac{p_r}{\hbar}{\mathcal N}_{N\ell}\sum_{t=0}^{N-\ell -1}(-1)^t\left( \begin{array}{c} N+\ell \\
N-\ell -1 -t\end{array}\right)\frac{1}{{t! (2\beta)^3}}(2x)^{(\ell +\frac{5}{2}+t)}
\nonumber\\
&\times& \sqrt{\frac{\beta \hbar\pi }{p_r}}\Gamma (\ell +3+t)\left(P^{-\frac{1}{2}}_{\ell +\frac{3}{2}+t}(x)
+i\frac{2}{\pi }Q_{\ell +\frac{3}{2} +t}^{-\frac{1}{2}
}(x)\right)\,,
\end{eqnarray*}
with $x$ defined in (\ref{z_argmnt}).
\end{lemma}
\begin{proof}
The statement follows by back substituting ${\mathcal I}_1\left( \frac{p_r}{2\hbar \beta}\right)$ and ${\mathcal I}_2\left(\frac{p_r}{2\hbar \beta} \right)$ in  \eqref{Ur7_a}  and \eqref{Ur8_a}, respectively, and taking into account \eqref{Ur7_pr}.
\end{proof}

Now, using the relations of $P_\nu^{-\mu}(x), Q_\nu^{-\mu}(x)$ with  $C_\gamma^\delta(x),D_\gamma^\delta(x)$, the Gegenbauer functions of the first and second kind respectively 
(see eq. (61) in  \cite{Durand}, and  eq.~(1.13) in \cite{HvH}, here transcribed in the notation of \cite{Durand}),

\begin{eqnarray*}
P_\nu^{-\mu}(x)&=&\frac{2^\mu}{\sqrt{\pi}}\frac{
\Gamma\left(\mu +\frac{1}{2}\right)\Gamma \left(\nu -\mu +1 \right)}
{\Gamma \left( \nu +\mu +1\right)}
(1-x^2)^{\frac{\mu }{2}}C^{\mu +\frac{1}{2}}_{\nu-\mu}(x)\,,
\\
Q_\nu^{-\mu}(x)&=& 2^{\mu }\frac{\sqrt{\pi }}{2}
\frac{\Gamma \left(\mu+\frac{1}{2}\right)\Gamma\left(\nu -\mu +1 \right)}
{\Gamma\left( \nu +\mu +1\right)} (1-x^2)^{\mu/2} D^{\mu +1/2}_{\nu -\mu}(x)\,,
\end{eqnarray*}
(in their respective intervals $x\in[-1,1]$, and $x\in (-1,1)$), we can rewrite
the expression for $\psi_{N\ell}(p_r)$. To this end, we take $\mu=1/2$, $\nu =\ell +t +3/2$
in the formulae above, to find

\begin{equation}\label{eqn_wtht_nmbr}
\psi_{N\ell}(p_r)=
\sum_{t=0}^{N-\ell -1}a^{(N)}_{\ell t}\, 
\sqrt{1-x^2}\, x^{(\ell +2 +t)}\left(  C^{1}_{\ell +1+t}(x)
+i  D^{1}_{\ell +1+t}(x)\right)
\end{equation}
where

\[
a^{(N)}_{ \ell t } ={\mathcal N}_{N\ell}2^{\ell +t +2}
\left( -1\right)^t \left( \begin{array}{c}N+\ell\\
N -\ell -1 -t\end{array}\right)\frac{1}{t!(2\beta)^2}
\Gamma(\ell +t+2)\,.
\]

Finally, by making a further change of variables,
\begin{eqnarray}
x&:=&\cos\gamma , \quad (1-x^2):=\sin^2\gamma ,
\label{trgnmtrc_nttns}
\end{eqnarray}
we arrive at the following result.

\begin{theorem}\label{thm2}
With the preceding notations, the radial momentum wave functions $\psi_{N\ell}(p_r)$ can  be expressed as the following finite series expansion:
\[
\psi_{N\ell}(p_r)=
\sum_{t=0}^{N-\ell -1}a^{(N)}_{ \ell t }  \sin \gamma \cos^{\ell +2 +t}\gamma 
\left(  C^{1}_{\ell +1+t}(\cos\gamma  )
+i D^{1}_{\ell +1+t}(\cos\gamma)\right)\,.
\]
\end{theorem}

Notice that this compact closed form of the radial momentum functions shows 
that, as in the case of the functions of Podolsky and Pauling, the functions
presented here are expressible  in terms of a $so(2,1)\sim su(1,1)$ basis.

As shown in \cite{HvH}, the combination, $C_\nu^\lambda(x)+iD_\nu^\lambda (x) $, defines a further Gegenbauer function of the second kind,  ${\mathcal D}_\nu^\lambda (x+i0)$, now of complex argument, by putting
\begin{eqnarray}
C_\nu^\lambda(x)+iD_\nu^\lambda(x)&=& 2{\mathcal D}_\nu^\lambda (x+i0),
\end{eqnarray}
where the notation $i0$, means the limit with respect to $\epsilon \to 0 $ in the
corresponding function with the variable $x+i\epsilon$.
With that, the expression for $\psi_{N\ell}(p_r)$ can be written even more concisely.

\begin{corollary}\label{cor1}
With the preceding notations, the radial momentum wave functions $\psi_{N\ell}(p_r)$ can  be expressed as
\begin{eqnarray}
&&\psi_{N\ell}(p_r)=
\sum_{t=0}^{N-\ell -1}2 a^{(N)}_{\ell t}\,\sqrt{1-x^2}\,  x^{(\ell +2 +t)}{\mathcal D}^1_{\ell +t +1}(x+i0),
\label{Ergnzng_1}
\end{eqnarray}
\end{corollary}

Notice that the use of Slater-type functions in (\ref{Ur7_a}) above,  
can cause numerical instabilities of the solutions (due to the alternating sign of the 
$(-1)^t$ factor) when the  summation extends to very high $(N-\ell +1)$ values, a critical  
remark already expressed in \cite{Weniger}. However, contrary to the case of the Fourier 
transforms of the radial wave functions considered in \cite{PdlPlng}, \cite{Weniger}, 
\cite{Yukcu}, in the present case we can not avoid this problem by invoking  
the generating functions of the Laguerre and Gegenbauer  polynomials. The reason for this is 
that, following this path, the hypergeometric function which would appear instead of $_2F_1$  in eq.~(20) in \cite{Weniger}, can not be transformed to the  binomial series (as $_1F_0$, in eq.~(21) in \cite{Weniger} does).

\subsection{Alternative finite complex  expansions in terms of elementary functions obtained  by means of Fourier sine and cosine transforms}
A simpler form for the radial momentum wave functions can be obtained by expressing the spherical Bessel function as  $j_0(x)=\sin x/x$  and converting
 $I_1\left( \frac{p_r}{2\hbar \beta }\right)$ in (\ref{Int_1st}) to,
\begin{eqnarray}
I_1\left(\frac{p_r}{2\hbar \beta} \right)&=&\int_0^\infty j_0\left( \frac{p_r}{2\hbar \beta }\rho \right)R_{N\ell }(\rho) r^2{\mathrm d}r\nonumber\\
&=&\frac{\hbar }{p_r}\frac{1}{(2\beta)^2}\int_0^\infty \sin \left( \frac{p_r}{2\hbar\beta  }\rho \right)R_{N\ell }(\rho) \rho {\mathrm d}\rho\,.\label{New_1}
\end{eqnarray}
Inserting above the expression for $R_{N\ell}(r)$ from 
(\ref{Ur1}), and taking into account (\ref{Ur4}), we get,

\begin{eqnarray}
I_1\left( \frac{p_r}{2\hbar \beta }\right)=\frac{\hbar }{p_r}\frac{1}{(2\beta)^2}&&\sum_{t=0}^{N-\ell -1} (-1)^t\frac{1}{t!}\left(\begin{array}{c}N+\ell \\ N -\ell -1 -t\end{array} \right)
\nonumber\\
&\times& \int_0^\infty \rho^{\ell+t+1} e^{-\frac{\rho}{2}}\sin  \left(\frac{p_r}{2\hbar\beta  }\rho \right)  {\mathrm d}\rho\, .
\label{New_2}
\end{eqnarray} 
Similarly, $n_0(x)$ in (\ref{Ur6}) can be expressed as $n_0(x)=-\cos x/x$, leading to
\begin{eqnarray*}
I_2\left(\frac{p_r}{2\hbar \beta } \right)=-\frac{\hbar }{p_r}\frac{1}{(2\beta)^2}&&\sum_{t=0}^{N-\ell -1} (-1)^t\frac{1}{t!}\left(\begin{array}{c}N+\ell\\ N-\ell -1 -t\end{array} \right)
\nonumber\\
&\times&\int_0^\infty \rho^{\ell+t+1} e^{-\frac{\rho}{2}}\cos \left(  \frac{p_r}{2\hbar\beta  }\rho\right)   {\mathrm d}\rho\,.
\end{eqnarray*} 
Integrals of the form (\ref{New_1}) and (\ref{New_2}), are standard Fourier  $\sin$ and $\cos$ transforms, which are well known  known  (see  {\bf 860.89} and {\bf 800.99} in \cite{Dwight}) and read:

\begin{eqnarray*}
\int_0^\infty x^ne^{-\alpha x}\sin bx\,  {\mathrm d}x
&=&\frac{\Gamma (n+1)\sin(n+1) \theta}{\left( \alpha^2 +b^2\right)^{\frac{n+1}{2}}},\\
\int_0^\infty x^ne^{-\alpha x}\cos bx\,  {\mathrm d}x
&=&\frac{\Gamma (n+1)\cos (n+1)\theta }{(\alpha ^2 +b^2)^{\frac{n+1}{2}}},
\\
\sin \theta =\frac{b}{ \left( \alpha ^2 +b^2\right)^{\frac{1}{2}}}, &\quad&
\cos\theta =\frac{\alpha }{ \left( \alpha ^2 +b^2\right)^{\frac{1}{2}}},\quad \theta=\tan^{-1}\frac{b}{\alpha}.
\end{eqnarray*}

\begin{theorem}\label{thm3}
With the preceding notations, the radial momentum wave functions $\psi_{N\ell}(p_r)$ can  be expressed as:
\begin{eqnarray}
\psi_{N\ell}(p_r)&=&\sum_{t=0}^{N-\ell -1}b^{(t)}_{N\ell}\frac{
e^{i(\ell +t +2) \tan^{-1}\left(\frac{p_r}{\hbar \beta }\right) }\left(\frac{1}{2}\right)^{\ell +t+2}}
{\left(
\sqrt{ \left( \frac{1}{2}\right)^2 + \left(
      \frac{p_r}{2\hbar \beta}\right)^2}\right)^ {\ell +t+2}},\label{prime} \\
&=&\sum_{t=0}^{N-\ell -1} b^{(t)}_{N\ell} e^{i(\ell +t +2) \theta }\cos^{\ell +t +2}\theta,\label{same_trg}
\end{eqnarray}
where we have introduced
\begin{eqnarray*}
b^{(t)}_{N\ell}&=&{\mathcal N}_{N\ell}2^{\ell +t +2}
\frac{1}{(2\beta)^2}(-1)^t\frac{1}{t!}\left(\begin{array}{c}N+\ell\\ N-\ell -1 -t\end{array} \right)\Gamma(\ell +t +2),.
\end{eqnarray*}
and
\[
\theta = \tan^{-1}\left(\frac{p_r}{\hbar \beta }\right)\,.
\]

\end{theorem}
\begin{proof}
Taking into account the preceding results, and making $n=\ell +t+1$, $ \alpha=1/2$, and $b=p_r/(2\hbar \beta)$, the function $\psi_{N\ell}(p_r)$ in (\ref{Ur7_pr}) takes the form,
\begin{eqnarray*}
\psi_{N\ell}(p_r)&=& {\mathcal N}_{N\ell}
\frac{1}{(2\beta)^2}\sum_{t=0}^{N-\ell -1} (-1)^t\frac{1}{t!}\binom{N+\ell}{N-\ell -1 -t}\Gamma(\ell +t +2)\nonumber\\
&\times &\left(\frac{\sin (\ell +t +2) \tan^{-1}\left(\frac{p_r}{\hbar \beta }\right)  }
{\left( \left( \frac{1}{2}\right)^2 + \left(\frac{p_r}{2\hbar \beta}\right)^2\right)^ {\frac{\ell +t+2}{2}}}
+i \frac{\cos (\ell +t +2) \tan^{-1}\left(\frac{p_r}{\hbar \beta } \right)
}
{\left( \left( \frac{1}{2}\right)^2 + \left(\frac{p_r}{2\hbar \beta}\right)^2\right)^ {\frac{\ell +t+2}{2}}}
\right)\,,
\end{eqnarray*}
from which the statement follows.
\end{proof}

The functions in (\ref{prime}), (\ref{same_trg}) are obviously equivalent to the \eqref{Ergnzng_1} in Corollary \ref{cor1}, because both are obtained from the same integrals (\ref{Int_1st}), and (\ref{Ur6}), just computed using a different technique. This is also visible from comparing their real and imaginary parts 
upon making use of relations
\begin{eqnarray}
C_n^1(\cos \theta )=\frac{\sin (n+1)\theta }{\sin \theta}, &\quad& 
D_n^1(\cos\theta) =\frac{\cos (n+1)\theta }{\sin \theta},
\quad \theta=\gamma,
\label{special_values}
\end{eqnarray}
for $n=\ell +t+1$. The latter are obtained upon substitution of the eqs.~(14.5.12) and (14.5.14) from \cite{NIST} into the expressions appearing in the above proof to Lemma 1. 

Two interesting properties are:
\begin{enumerate}
\item The functions are continuous,
finite at the origin, and approach zero at infinity. Therefore, they are bounded continuous and
hence square integrable. 
\item The denominator in (\ref{prime})  can be rewritten using,
\[
 \left( \frac{1}{2}\right)^2 +\left( \frac{p_r}{2\hbar \beta}\right)^2=\frac{1}{(2\hbar \beta)^2 }\left(
\hbar^2\beta^2 +p_r^2\right)=\frac{1}{(2\hbar \beta)^2 }\left(
  \frac{{\mathcal E}_N^2}{c^2} +p_r^2\right)\,,
\]
(where ${\mathcal E}_N^2=-2\mu c^2  E_N $). This is a manifestly
$SO(4)-$symmetric expression in an Euclidean four-dimensional space, in which $p_r$ is interpreted as the radius of a two-dimensional sphere, while ${\mathcal E}_N^2/c^2+p_r^2=s_N^2$ could be viewed as the squared  radius of a three-dimensional sphere. Moreover, recalling that 
\begin{eqnarray}
\tan^{-1}\frac{p_r}{\hbar \beta }&=&\cos^{-1} \frac{\left(\frac{1}{2}\right)}{\sqrt{
\left( \frac{1}{2}\right)^2 +\left(\frac{p_r}{2\hbar \beta } \right)^2}}=\theta,
\label{Knall}
\end{eqnarray} 
it becomes apparent that, as was the case of the Podolsky-Pauling wave functions, the functions in \eqref{same_trg} are $SO(4)-$symmetric.
\end{enumerate}

In particular, the second property shows that the well known $SO(4)$ symmetry of the Hydrogen
atom in the position representation, is also present in the momentum representation, which is
a consistency test. 

\subsection{Equivalence with the functions of Lombardi-Ogilvie}

As mentioned earlier, the greatest drawback of Lombardi-Ogilvie functions \eqref{Olg_33}
is that they are not manifestly $SO(4)-$invariant. Our functions \eqref{prime}, \eqref{same_trg}, \eqref{Ergnzng_1} do present that invariance, and now we are going to
show that they are equivalent to those in \eqref{Olg_33}, thus curing that problem.

We start by considering again the radial Schr\"odinger equation (we work in $\hbar =1$ units):
\begin{equation*}
\left[\frac{\hat{p}^2_r}{2\mu}+\frac{\ell (\ell+1)}{2\mu r^2}-\frac{Ze^2}{r}\right] R_{N\ell}(r)
=E_NR_{N\ell}(r)\,.
\end{equation*}
If we apply the integral transform $\mathcal{H}$ in Theorem \ref{theorem1} to this equation,
and write $\psi_{N\ell}(p_r)$ for the transform of $R_{N\ell}(r)$, we arrive at
\begin{equation}\label{equivalence01}
(p_r^2-2\mu E)\psi_{N\ell}(p_r) -2\mu Ze^2 \int_0 ^{\infty }\frac{e^{ip_rr}}{r}\frac{R_{N\ell }(r)}{r}r^2 {\mathrm d}r
+\ell(\ell +1) \int_0 ^{\infty }\frac{e^{ip_rr}}{r}\frac{R_{N\ell }(r)}{r^2}
r^2 {\mathrm d}r =0\,.
\end{equation}

We want to show that this equation is equivalent to the defining equation of the 
Lombardi-Ogilvie functions \eqref{Lmbrd1}, repeated here for the reader's 
convenience:
\[
{\Big(}p_r^2-2\mu E  +2\mu Ze^2 I +\ell (\ell+1)I^2  {\Big)}\alpha_{N\ell}(p_r)=0\,,
\]
where, recall, $I$ denotes the integral operator\footnote{For an analysis of this operator, see
\cite{Magyari}, p. 134, Problem 19.} in \eqref{2nd_rplsmnt},
\[
I=-i\int_{-\infty}^{p_r} {\mathrm d}{p^\prime _r}\,.
\]

In their setting, Lombardi and Ogilvie do not specify the domain of this operator, but for
sufficiently regular functions (in particular, for those in the dense subspace
$\mathcal{C}^\infty_0(\mathbb{R})$),
we can take derivatives with respect to $p_r$ in \eqref{equivalence01}, and evaluate the 
following integrals:

\begin{eqnarray}
-i\frac{{\mathrm d}}{{\mathrm d}p_r}\int_0 ^{\infty }\frac{e^{ip_rr}}{r}\frac{1}{r}R_{N\ell }(r)r^2 {\mathrm d}r&=&
\int_0 ^{\infty }\frac{e^{ip_rr}}{r}R_{N\ell }(r)r^2 {\mathrm d}r\nonumber=\psi_{N\ell}(p_r)\\
&=& -i \frac{{\mathrm d}}{{\mathrm d}p_r }\int _{-\infty}^{p_r}{\mathrm d}p_r^\prime 
\alpha_{N\ell}(p_r^\prime ) = \alpha_{N\ell}(p_r),\nonumber\\
\label{waw}
\end{eqnarray}
and
\begin{eqnarray}
&&\left( -i \frac{{\mathrm d}}{{\mathrm d}p_r}\right)
\left( -i \frac{{\mathrm d}}{{\mathrm d}p_r}\right)
\int_0^\infty \frac{e^{ip_rr}}{r}
\frac{1}{r^2}R_{N\ell}(r) r^2{\mathrm d}r=\int_0^\infty 
\frac{e^{ p_r r}}{r}R_{N\ell }(r) r^2 {\mathrm d}r\nonumber\\
&=&\psi_{N\ell}(p_r)=\left(-i \frac{{\mathrm d} }{{\mathrm d}p_r} \right)
\left(-i \frac{{\mathrm d} }{{\mathrm d}p_r^\prime } \right)
\int_{-\infty }^{p_r}
{\mathrm d}p_r^\prime \int_{-\infty }^{p_r^\prime}
{\mathrm d}p_r^{\prime \prime} 
\alpha_{N\ell }(p_r^{\prime \prime})\nonumber\\
&=& \alpha_{N\ell}(p_r),
\label{waw_m}
\end{eqnarray}
which follow from the definition of $\psi_{N\ell}(p_r)$ as an integral transform, in 
combination with the fundamental theorem of the integral calculus.
The above considerations show that our method has the advantage over eq.~(\ref{Lmbrd1})
of being generalizable to any radial potential and any dimensions, in which case one will encounter different integrals, possibly multidimensional ones.

\subsection{Comparison between the functions of Podolsky-Pauling and Lombardi-Ogilvie}\label{sec44}
A comparison between the  momentum space wave functions of Podolsky and Pauling, on the one side, with those of Lombardi and Ogilvie, on the other, is in order. For the sake of transparency, we here limit ourselves to compare the corresponding profiles of the momentum distributions, and ignore normalization constants.
Also, for illustrative purposes, we select some quantum numbers, for which the expressions look particularly simple. To be specific, below we compare the respective momentum distributions for $\ell= N-1$.

Modulo normalization constants, the Podolsky-Pauling distributions for maximal $\ell$ values  corresponding to (\ref{Szmytk}), read:
\begin{equation}
|{\mathcal G}_{N, N-1}(p)|^2\sim \frac{\left(4\hbar \beta p\right)^{2(N-1)}}
{\left(\hbar^2\beta^2 +p^2\right)^{2(N+1)}}
\label{eqtn_1}
\end{equation}
For comparison, those of Lombardi-Ogilvie, in our parametrization (\ref{same_trg}) are given by,
\begin{equation}
|\psi_{N, N-1}(p_r)|^2\sim 
\frac{1}{\left(\hbar^2\beta^2+p_r^2\right)^{N+1}}
\label{eqnt_2}
\end{equation}

Graphically:
\begin{figure}[H]
\centering
\begin{subfigure}{.5\textwidth}
    \centering
    \includegraphics[width=0.8\textwidth]{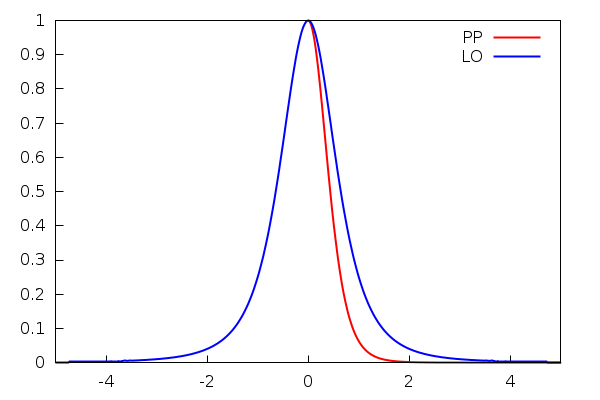}
\end{subfigure}%
\begin{subfigure}{.5\textwidth}
    \centering
    \includegraphics[width=0.8\textwidth]{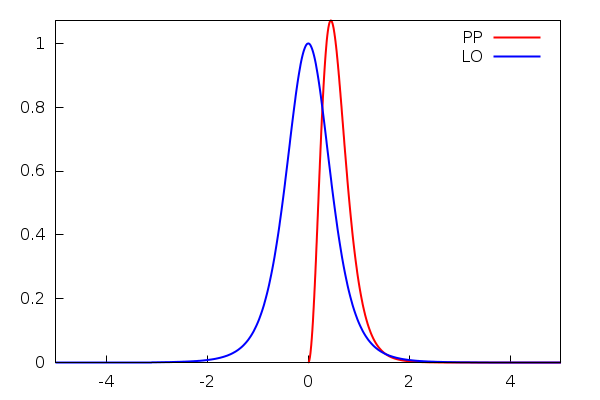}
\end{subfigure}
\begin{subfigure}{.5\textwidth}
    \centering
    \includegraphics[width=0.8\textwidth]{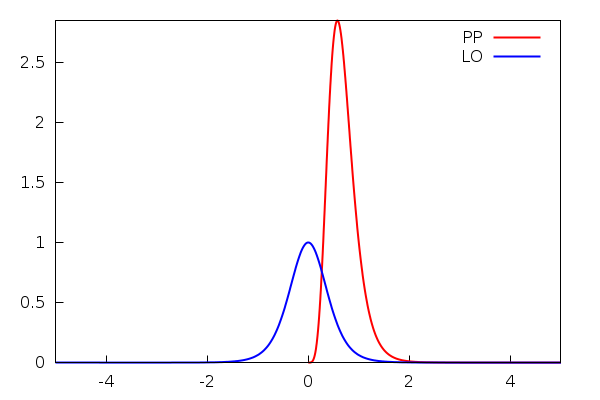}
\end{subfigure}%
\begin{subfigure}{.5\textwidth}
    \centering
    \includegraphics[width=0.8\textwidth]{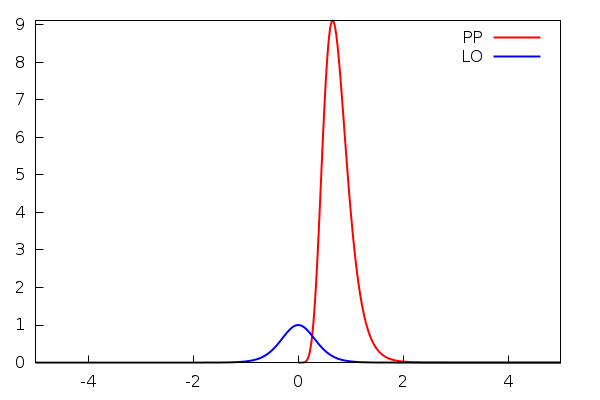}
\end{subfigure}
\caption[short]{Plots of the Podolsky-Pauling (PP) and Lombardi-Ogilvie (LO) non-normalized
momentum distributions, for $N=1,2,3,4$ (in order, left to right and top to bottom), in units 
$\hbar \beta =1$. The argument of 
the former is $p\in \,]0,\infty[\,$, while that of the latter is $p_r\in \,]-\infty,\infty[\,$.}
\end{figure}

According to (\ref{eqtn_1}), the Podolsky-Pauling momentum distribution for the ground state falls off as $ (1+p^2/\hbar^2\beta^2)^{-4}$, a fact experimentally confirmed by \cite{Lohmann}. 
This can be explained by noticing that in the space spanned by the radial Hydrogen atom wave functions and their Fourier transforms to momentum space the following uncertainty relation holds valid:
\begin{eqnarray}
\langle r^2\rangle \langle p^2\rangle \geq \frac{D^2}{4},
\label{Virial}
\end{eqnarray} 
where $D$ is the dimensionality of the space, $D=3$ in our case (see \cite{Gadre})\footnote{Alternatively, such Heisenberg-like uncertainty relations can also been confirmed and generalized on the basis of  information-theoretic methods \cite{Angulo1}, \cite{Angulo2}.}. The inequality in (\ref{Virial})  encodes the complementary character of $r^2$ and $p^2$, and provides the explanation to the fact that the wave functions of Podolsky and Pauling are the ones appropriate for the description  of the $p^2$ expectation values  distribution in the Hydrogen atom states. However, if what we want is the distribution of the expectation values of the radial momentum, $\langle p_r\rangle$, then the Lombardi-Ogilvie wave functions should be more appropriate. 
Indeed, our conclusion is that the two sets of functions under discussion are complementary to each other. While the ones by Podolsky and Pauling test the distribution of linear momenta with arbitrary orientation, those of Lombardi and Ogilvie would test the distribution of strictly  radially aligned  momenta, and this could be experimentally confirmed by appropriate radial polarization diffraction experiments, along the lines of \cite{TMa}.

\section{Conclusions and perspectives}\label{sec5}

We propose a construction of the Hydrogen atom radial wave functions in momentum space from
an integral transform which generalizes the case of the free particle momentum representation 
(see Theorem \ref{theorem1}). The corresponding kernel is given by a spherical wave, instead of 
the traditional plane one, which comes from the consideration of the \emph{formal}
eigenfunctions of the symmetric radial momentum operator \eqref{eq07}. This operator,
although not self-adjoint, has the feature of satisfying Heisenberg's commutation rules with 
the radial coordinate.

The resulting radial wave functions are equivalent to those of Lombardi and Ogilvie 
\cite{Lombardi,Ogilvie}, and can be expressed either in terms of finite complex 
expansions of Gegenbauer functions of the first and second kind (Theorem \ref{thm2} 
and Corollary \ref{cor1}), or in terms of (elementary) trigonometric functions (Theorem 
\ref{thm3} in Section \ref{sec4}). The first representation reflects the $so(2,1)\sim su(1,1)$ Lie algebra symmetry 
properties of the momentum wave functions, while the second one shows explicitly properties
such as square integrability and $SO(4)$ symmetry, a feature that was absent from the 
Lombardi-Ogilvie functions.

The total wave functions of the Hydrogen atom in momentum space, if needed, can be obtained as 
products of  the radial functions presented here with the angular functions either of Lombardi 
and Ogilvie \cite{Ogilvie}, or with those of Liu, Tang and Xun in \cite{QHLiu}, giving
qualitatively different outcomes. Which is the correct one is a question that must be decided
by means of comparison with experimental data.

The wave functions expressed in momentum representation, could be of use in studying the radial 
diffraction of hydrogen beams  through lattices such as graphene \cite{Brand}, roaming of 
Hydrogen atoms in the decomposition of organic compounds  \cite{Roaming}, momentum 
probabilities for particles in a  Coulomb
well \cite{Riggs}, tunneling \cite{Morrone}, \cite{Jaffe}, and other phenomena, in which  
Heisenberg's uncertainty relation between the radial coordinate and the radial momentum is required, among them  possibly  questions of information entropies  \cite{Majernik}.

In addition, our approach may also be of interest to the construction of  Wigner's 
phase-space functions and their applications, for example, along the lines of \cite{Praxmeyer}.

\end{document}